\documentclass[letterpaper, 10 pt, conference]{ieeeconf}
\title{Aerial_Robots_CDC_2017_Arxiv_Version}                                      
\IEEEoverridecommandlockouts                                                                       
\usepackage[dvipsnames]{xcolor}
\usepackage{amsmath}
\usepackage{mathtools}
\usepackage{amsmath}
\usepackage{amssymb}
\usepackage{tikz}
\usepackage{nicefrac}
\usepackage{mathtools, cuted}
\usepackage{lipsum, color}
\usepackage{float}
\usepackage{afterpage}
\usepackage{graphicx}
\usepackage{caption}
\usepackage{subfig}
\usepackage{tabularx}
\usepackage{cite}
\usepackage{balance}

\DeclareMathOperator*{\Min}{min}
\newtheorem{theorem}{Theorem}
\newtheorem{lemma}{Lemma}
\newtheorem{assumption}{Assumption}

\title{\LARGE \bf
Joint optimization of transmission and propulsion in aerial communication networks
}

\author{Omar J. Faqir, Eric C. Kerrigan, and Deniz G\"und\"uz
\thanks{The support of the EPSRC Centre for Doctoral Training in High Performance Embedded and Distributed Systems  (HiPEDS, Grant Reference EP/L016796/1) is gratefully acknowledged.}
\thanks{O. J. Faqir and Deniz G\"und\"uz are with the Department of Electrical \& Electronics Engineering, Imperial College London, SW7~2AZ, U.K. {\tt\small ojf12@ic.ac.uk}, {\tt\small d.gunduz@ic.ac.uk}}%
\thanks{Eric C. Kerrigan is with the  Department of Electrical \& Electronic Engineering
and Department of Aeronautics, Imperial College London, London SW7~2AZ, U.K. {\tt\small e.kerrigan@imperial.ac.uk}}%
}

\begin{document}
\maketitle
\thispagestyle{empty}
\pagestyle{empty}

\begin{abstract}
	
Communication energy in a wireless network of mobile autonomous agents should be considered as the sum of transmission energy and propulsion energy used to facilitate the transfer of information.
Accordingly, communication-theoretic and Newtonian dynamic models are developed to model the communication and locomotion expenditures of each node. 
These are subsequently used to formulate a novel nonlinear optimal control problem (OCP) over a network of autonomous nodes. 
It is then shown  that, under certain conditions, the OCP can be transformed into an equivalent convex form.
Numerical results for a single link between a node and access point allow for comparison with known solutions before the framework is applied to a multiple-node UAV network, for which previous results are not readily extended.
Simulations show that transmission energy can be of the same order of magnitude as propulsion energy allowing for possible savings, whilst also exemplifying how speed adaptations together with power control may increase the network throughput.

\end{abstract}

\section{Introduction}

We aim to derive a control strategy to minimize communication energy in robotic networks. In particular, uninhabited aerial vehicle (UAV) networks are considered, with results being generalizable to broader classes of autonomous networks. A dynamic transmission model, based on physical layer communication-theoretic bounds, and a mobility model for each node is considered alongside a possible network topology. As a cost function, we employ the underused interpretation of communication energy as the sum of transmission energy and propulsion energy used for transmission, i.e.\ when a node changes position to achieve a better channel.

For simulation purposes we consider the two wireless network setups shown in Figure~\ref{fig:TwoNodeDiagram}. We first present the most basic scenario consisting of a single agent $U_1$ moving along a predefined linear path while offloading its data to a stationary access point~(AP). We compare results for variable and fixed speeds, before studying a two-agent single-hop network.

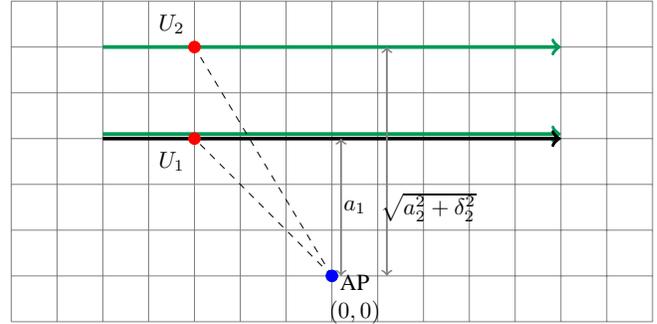
\begin{figure}[t!]
\centering{
\resizebox{\linewidth}{!}{
\begin{tikzpicture}[scale=0.7]
	\draw[step=1cm,gray,ultra thin] (-2,-1) grid (12,6);

	\draw[ForestGreen, ultra thick,->] (0,5) -- (10,5);
   	\draw[ForestGreen, ultra thick,->] (0,3.1) -- (10,3.1);
	\draw[black, ultra thick,->] (0,3) -- (10,3);

	\draw[black,dashed] (5,0) -- (2,3);
	\draw[black,dashed] (5,0) -- (2,5);
	\draw[gray, thick, <->] (5.2,0)--(5.2,3);
	\draw[gray, thick, <->] (6.2,0)--(6.2,5);

	\draw[red,fill=blue,blue] (5,0) circle (.8ex);
   	\draw[red,fill=red,red] (2,3) circle (.8ex);
	\draw[red,fill=red,red] (2,5) circle (.8ex);

    \node[draw=none, align=center] at (1.5,2.5) {$U_1$};
    \node[draw=none, align=center] at (1.5,5.5) {$U_2$};
    \node[draw=none, align=center] at (5.5,-0.5) {AP \\ $(0,0)$};
    \node[draw=none, align=center] at (5.5,1.5) {$a_1$};
    \node[draw=none, align=center] at (7.1,1.5) {$\sqrt{a_2^2+\delta_2^2}$};
\end{tikzpicture}}}
\caption{Geometric configuration for simulation setups featuring $N=1$ (black) and $N=2$ (green) nodes. Speeds along these paths may be variable or fixed.
The altitudes and lateral displacements of $U_1,U_2$ are $a_1=a_2=1000$\,m and $\delta_1=0,\delta_2=1000$\,m, respectively.}
\label{fig:TwoNodeDiagram}
\end{figure}

For UAV networks, research efforts largely break down into two streams: the use of UAVs in objective based missions (e.g.\ search and pursuit \cite{ko2002network}, information gathering/mobile sensor networks \cite{wang2013robotic,thammawichai2016optimizing}), and use as supplementary network links~\cite{zhan2011wireless}.
Optimal completion of these macro goals has been addressed in the literature, but there is no necessary equivalence between optimal task-based and energy-efficient operations. 

Efforts concerning mobility focus on mobile (in which node mobility models are random) or vehicular (where mobility is determined by higher level objectives and infrastructure) ad-hoc networks \cite{bekmezci2013flying}.
Since neither are fully autonomous networks, mobility is not available as a decision variable. The work in\cite{zhao2003message} introduced the concept of proactive networks, where certain nodes are available as mobile relays. 
However, the focus is on relay trajectory design and a simplistic transmission model is assumed, inherently prohibiting energy efficiency. The related problem of router formation is investigated in \cite{yan2012robotic} using realistic models of communication environments.

We assume hard path constraints, possibly due to the existence of higher level macro objectives, but allow changes in trajectory along the path by optimizing their speed (as in~\cite{Sujit2014unmanned}, we define a trajectory as being a time-parameterized path). 
Use of fixed paths does not restrict our results as most UAV path planning algorithms operate over longer time horizons and are generally restricted to linear or circular loiter trajectories \cite{Sujit2014unmanned}.
A linear program (LP) is used in \cite{yan2014go} to determine how close a rolling-robot should move before transmission in order to minimize total energy. However, the linear motion dynamics used restricts applicability of the model. Similarly to our current work, \cite{zeng2016throughput} uses a single mobile UAV relay to maximize data throughput between a stationary source-destination pair. An optimal trajectory for an a priori transmission scheme is iteratively found. Similarly, for a given trajectory, the optimal relaying scheme may be obtained through water-filling over the source-to-relay and relay-to-receiver channels.

Our contribution differs from the above works in terms of the formulation of a more general nonlinear convex OCP for finding joint transmission and mobility strategies to minimize communication energy. We solve this problem, exemplifying possible savings for even just a single node. As a final point, we show analytically and numerically that, even at fixed speeds, the optimal transmission scheme for a two-user multiple-access channel(MAC) is counter-intuitive and not captured by na\"{\i}ve transmission policies. 

\section{Problem Description}
Consider $N$ homogeneous mobile nodes $U_n, n \in \mathcal{N} \triangleq \{1,\ldots,N\}$, traveling along linear non-intersecting trajectories at constant altitudes $a_n$  and lateral displacements $\delta_n$ over a time interval $\mathcal{T}\triangleq[0,T]$. The trajectory of node $U_n$ is denoted by $t \mapsto (q_n(t),\delta_n,a_n)$, relative to a single stationary AP $U_0$ at position $(0,0,0)$ in a three dimensional space. 
At $t=0$, $U_n$ is initialized with a data load of $D_n$ bits, which must all be offloaded to $U_0$ by time $t=T$. We consider a cooperative network model, in which all  nodes cooperate to offload all the data in the network to the AP by relaying each other's data. Each node has a data buffer of capacity $M$ bits, which limits the amount of data it can store and relay. 

\subsection{Communication Model}
We employ scalar additive white Gaussian noise (AWGN) channels. For UAV applications, we assume all links are dominated by line-of-sight (LoS) components, resulting in flat fading channels, meaning
all signal components undergo similar amplitude gains~\cite{tse2005fundamentals}. All nodes have perfect information regarding link status, which in practice may be achieved through feedback of channel measurements, while the overhead due to channel state feedback is ignored. 

Similar to \cite{ren2004improved}, for a given link from source node $U_{n}$ to receiver node $U_m$, the channel gain $\eta_{nm}(\cdot)$ is expressed as
\begin{equation}
\eta_{nm}(q_{nm}) \triangleq \frac{G}{\left(\sqrt[]{a_{nm}^2 + \delta_{nm}^2+ q_{nm}^2}\right)^{2\alpha}},
\end{equation}
where $q_{nm} \triangleq q_n - q_m$, constant $G$ represents transmit and receive antenna gains and $\alpha \geq 1$ the path loss exponent. We define $a_{nm}$ and $\delta_{nm}$ in a similar fashion. The channel gain is inversely related to the Euclidean distance between nodes.

Each node has a single omnidirectional antenna of maximum transmit power of $P_{\text{max}}$ Watts. We consider half duplex radios; each node transmits and receives over orthogonal frequency bands. Accordingly, a different frequency band is assigned for each node's reception, and all messages destined for this node are transmitted over this band, forming a MAC. We do not allow any coding (e.g.\ network coding) or combining of different data packets at the nodes, and instead consider a decode-and-forward-based routing protocol at the relay nodes \cite{gunduz2007opportunistic}. The resulting network is a composition of Gaussian MACs, for each of which the set of achievable rate tuples defines a polymatroid capacity region
\cite{tse1998multiaccess}. If~$N$ nodes simultaneously transmit independent information to the same receiver, the received signal is a superposition of the transmitted signals scaled by their respective channel gains, plus an AWGN term. We model the achievable data rates using Shannon capacity, which is a commonly used upper bound on the practically achievable data rates subject to average power constraints. Due to the convexity of the capacity region, throughput maximization does not require time-sharing between nodes \cite{tse1998multiaccess}, but may be achieved through successive interference cancellation (SIC).

Consider a single MAC consisting of $N$ users $U_n, n \in \mathcal{N}$, transmitting to a receiver $U_m, m \not\in \mathcal{N}$. The capacity region $\mathcal{C}_N(\cdot,\cdot)$, which denotes the set of all achievable rate tuples~$r$, is defined as
\begin{equation}
\mathcal{C}_N(q, p) \triangleq \left\{r \geq 0 \mid  f_m(q,p,r,\mathcal{S}) \leq 0, \forall \mathcal{S} \subseteq \mathcal{N} \right\},
\end{equation}
where $q$ is the tuple of the differences $q_{nm}$ in positions between the $N$ users and the receiver, $p \in \mathcal{P}^N$ is the tuple of transmission powers allocated by the $N$ users on this channel, and $\mathcal{P} \triangleq [0,P_{\text{max}}]$ is the range of possible transmission powers for each user.
$f_m(\cdot)$ is a nonlinear function bounding $\mathcal{C}_N(q, p)$, given by 
\begin{multline}
f_m(q,p,r,\mathcal{S}) \triangleq \ \sum_{n \in\mathcal{S}}r_n -\\B_m\log_2\left(1 + \sum_{n \in S}\frac{\eta_{nm}(q_{nm})p_n}{ \sigma^2} \right),
\end{multline}

where $r_n$ is the $n^{\text{th}}$ component of $r$, $B_m$ is the bandwidth allocated to $U_m$, and $\sigma^2$ is the receiver noise power.
Consider the example (Section~\ref{sec:MultNode})  where we do not allow relaying. This gives rise to a MAC with  $N=2$ transmitters $U_1,U_2$ and the AP $U_0$. The capacity region $\mathcal{C}_2(q,p)$ is the set of non-negative tuples $(r_1,r_2)$ that satisfy
\begin{subequations}
\label{eq:CapRegion2}
\begin{align}  
r_1  &\leq \textstyle B_0\log_{2}\left(1+ \dfrac{\eta_{10}(q_{10})p_1}{\sigma^2}\right) \\
r_2  &\leq \textstyle B_0\log_{2}\left(1+ \dfrac{\eta_{20}(q_{20})p_2}{\sigma^2}\right) \\
r_1 + r_2 &\leq \textstyle B_0\log_{2}\left(1+ \dfrac{\eta_{10}(q_{10})p_1+\eta_{20}(q_{20})p_2}{\sigma^2}\right) 
\end{align}
\end{subequations}
for all $(p_1,p_2) \in \mathcal{P}^2$. The first two bounds restrict individual user rates to the single-user Shannon capacity. Dependence between  $U_1$ and $U_2$ leads to the final constraint, that the sum rate may not exceed the point-to-point capacity with full cooperation. For transmit powers $(p_1,  p_2)$  these constraints trace out the pentagon shown in Figure~\ref{fig:CapacityRegion}. The sum rate is maximized at any point on the segment~$L_3$. Referring to SIC, the rate pair at boundary point $R^{(1)}$ is achieved if the signal from source $U_2$ is decoded entirely before source~$U_1$, resulting in the signal from $U_2$ being decoded at a higher interference rate than the signal from $U_1$. At $R^{(2)}$ the opposite occurs.

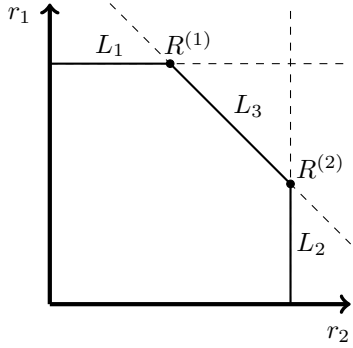
\begin{figure}[t!]
\centering
\vspace{0.5em}
\begin{tikzpicture}[scale=0.8]
	\draw[black, ultra thick, ->] (0,0) -- (0,5);
	\draw[black, ultra thick, ->] (0,0) -- (5,0);

	\draw[black,dashed] (0,4) -- (5,4);
	\draw[black,dashed] (4,0) -- (4,5);
	\draw[black,dashed] (1,5) -- (5,1);

	\draw[black, thick] (0,4) -- (2,4);
	\draw[black, thick] (4,0) -- (4,2);
	\draw[black, thick] (2,4) -- (4,2);

    \draw[black,fill=black,black] (4,2) circle (.4ex);
   	\draw[black,fill=black,black] (2,4) circle (.4ex);


	\node[draw=none, align=center] at (-0.5,4.8) {$r_1$};
    \node[draw=none, align=center] at (4.8,-0.5) {$r_2$};
    \node[draw=none, align=center] at (2.3,4.35) {$R^{(1)}$};
    \node[draw=none, align=center] at (4.5,2.3) {$R^{(2)}$};
    \node[draw=none, align=center] at (1,4.3) {$L_1$};
    \node[draw=none, align=center] at (4.35,1) {$L_2$};
    \node[draw=none, align=center] at (3.3,3.3) {$L_3$};

\end{tikzpicture}
    \caption[Two user capacity region]{Capacity region for a given power policy across two parallel channels, with corner rate pairs labeled as $R^{(1)}=(r_1^{(1)},r_2^{(1)})$ and $R^{(2)}=(r_1^{(2)},r_2^{(2)})$ and line segments labeled as $L_1, L_2, L_3$.}
    \label{fig:CapacityRegion}
\end{figure}

\subsection{Propulsion Energy Model}

The electrical energy used for propulsion in rolling robots has been modeled as a linear or polynomial function of speed in \cite{yan2014go,mei2004energy} respectively.  We take a more general approach, restricting the fixed wing UAV to moving at strictly positive speeds and using Newtonian laws as a basis, as in \cite{ali2016motion}. The function $\Omega(\cdot)$ models the resistive forces acting on node $U_n$ in accordance with the following assumption.

\begin{assumption} \label{ass:umption2}
The resistive forces acting on each node~$U_n$ may be modeled by the function $x \mapsto \Omega(x)$ such that $x \mapsto x\Omega(x)$ is convex on $x \in [0,\infty)$ and $\infty$ on $x \in (-\infty,0)$.
\end{assumption}

Comparatively, in the fixed wing model proposed in \cite{zeng2016energy}, the drag force of a UAV traveling at constant altitude at sub-sonic speed $v$ is 
\begin{equation}
\Omega(v) = \frac{\rho C_{D0}Sv^2}{2} + \frac{2L^2}{(\pi e_0 A_R)\rho S v^2}
\end{equation}
where the first term represents parasitic drag and the second term lift-induced drag. Parasitic drag is proportional to $v^2$, where $\rho$ is air density, $C_{D0}$ is the base drag coefficient, and~$S$ is the wing area. Lift induced drag is proportional to $v^{-2}$, where $e_0$ is the Oswald efficiency, $A_R$  the wing aspect ratio and $L$ the induced lift \cite{zeng2016energy}. For fixed-altitude flight, $L$ must be equal to the weight of the craft $W=mg$. The power required to combat drag is  the product of speed and force. 

The propulsion force $F_n(\cdot)$ must satisfy the force balance equation 
\begin{equation}
F_n(t) - \Omega(v_n(t)) = m_n\dot{v}_n(t),
\end{equation}
where $m_n$ is the node mass, $v_n(t)$ is the speed and~$\dot{v}_n(t)$ is the acceleration. The instantaneous power used for propulsion is the product $v_n(t)F_n(t)$, with the total propulsion energy taken as the integral of this power over $\mathcal{T}$. We assume $v_n(t) \geq 0, \ \forall t \in \mathcal{T}$, which is valid for fixed wing aircrafts. Thrust is restricted to the range $[F_{\text{min}},F_{\text{max}}]$.

\subsection{General Continuous-Time Problem Formulation}
We formulate the problem in continuous-time. 
At time~$t$, node $U_n, n \in \mathcal{N}$ can transmit to any node $U_m, m \in \{\mathcal{N},N+1\}\setminus\{n\}$ at a non-negative data rate $r_{nm}(t)$ using transmission power $p_{nm}(t)$. The sum power used in all outgoing transmissions from $U_n$ is denoted by $p_n(t)$. From this, the set of achievable data rates is bounded above by a set of $2^{|\mathcal{N}|}-1$ nonlinear submodular functions $f_m(\cdot,\cdot,\cdot,\cdot)$, where $|\cdot|$ applied to a set denotes the cardinality operator. 
Exponential growth in the number of nodes is a computational intractability.
Hence, results are limited to small or structured networks where only a subset of nodes use each MAC.

The trajectory of node $U_n$ is denoted by the tuple 
\begin{equation}
Y_n \triangleq (p_n,r_n,s_n,q_n,v_n,\dot{v}_n,F_n),
\end{equation}
where $q_n(t)$ is the node's position at time $t$ and $s_n(t)$ the state of its storage buffer subject to maximum memory of $M$  bits. The optimal control problem that we want to solve is
\begin{subequations} \label{eq:GeneralProbForm}
\begin{align}
& \Min_{p,r,s,q,v,F}
\sum_{n=1}^{N} \int_{0}^{T} p_n(t) + v_n(t) F_n(t) \mathrm{d}t \label{eq:CostFunc} \\
   \text{s.t. } & \forall n \in \mathcal{N},  m \in \{\mathcal{N},N+1\},  t \in \mathcal{T},\mathcal{S}\subseteq\mathcal{N} \notag 
\end{align}  
\begin{align}
& f_m(q(t),p(t),r(t),\mathcal{S}\setminus\{m\}) \leq 0 \label{eq:CapRegionConst} \\
& \dot{s}_n(t) =   \sum_{m \neq n}^{N}  r_{mn}(t) - \sum_{m\neq n}^{N+1} r_{nm}(t) \label{eq:storageUpdateConst}  \\
& s_n(0) = D_{n} , \quad s_n(T) = 0 \label{eq:initAndFinalStorageConst} \\
& q_n(0) = Q_{n,\text{init}}, \quad q_n(T) = Q_{n,\text{final}} \label{eq:initAndFinalPosConst} \\
& v_n(0) = v_{n,\text{init}} \label{eq:initVelocity} \\
&F_n(t)  = m_n\dot{v}_n(t) + \Omega(v_n(t)) \label{eq:AccelerationBound} \\
&\dot{q}_n(t) = \zeta_n v_n(t) \label{eq:defVel} \\
& Y_{n,\text{min}} \leq Y_n(t) \leq Y_{n,\text{max}}    \label{eq:varBounds} 
\end{align}
\end{subequations}
The cost function (\ref{eq:CostFunc}) is the sum of nodal transmission and propulsion energies. Constraint (\ref{eq:CapRegionConst}) bounds the achievable data rate to within the receiving nodes' capacity region, and~\eqref{eq:storageUpdateConst} updates the storage buffers with sent/received data. Constraints (\ref{eq:initAndFinalStorageConst}) act as initial and final constraints on the buffers, while (\ref{eq:initAndFinalPosConst})--(\ref{eq:defVel}) ensure all nodes travel from their initial to final destinations without violating a Newtonian force-acceleration constraint; $\zeta_n\in\{-1,1\}$ depending on whether the position $q_n(t)$ decreases or increases, respectively, if the speed $v_n(t)\geq 0$. The final constraint (\ref{eq:varBounds}) places simple bounds on the decision variables, given by
\begin{subequations}\begin{align}
Y_{n,\text{min}} &\triangleq (0,0,0,
-\infty,
V_{\text{min}},-\infty,F_{\text{min}}), \\
Y_{n,\text{max}} &\triangleq (P_{\text{max}},\infty,M,
\infty,
V_{\text{max}},\infty,F_{\text{max}}),
\end{align}
\end{subequations}
where $0 \leq V_{\text{min}} \leq V_{\text{max}}$ 
and $F_{\text{min}} \leq F_{\text{max}}$. The above optimal control problem may then be fully discretized using optimal control solvers, such as ICLOCS \cite{falugi2010imperial}. Before simulation results are presented we prove that this problem admits an equivalent convex form under certain conditions.

\section{Convexity Analysis}
Efficient convex programming methods exist, which may be used in real-time applications. We first show that the nonlinear data rate constraints (\ref{eq:CapRegionConst}) are convex in both positions and transmission power. We then show that the non-linear equality constraint (\ref{eq:AccelerationBound}) may be substituted into the cost function, convexifying the cost function. This, however, turns the previously simple thrust bound $F_\text{min} \leq F_n(t)$ into a concave constraint, resulting in a convex OCP if thrust bounds are relaxed. 
The absence of thrust bounds arises when considering a fixed trajectory, or is a reasonable assumption if the speed range is sufficiently small.

\begin{lemma} \label{lemma:1}
The rate constraints (\ref{eq:CapRegionConst}) are convex in powers and positions for all path loss exponents $\alpha \geq 1$.
\end{lemma}
\begin{proof}
By writing the channel gains as an explicit function of node positions, for receiver $U_m$ each of the capacity region constraints is of the form
\begin{multline}
\sum_{n \in S}r_{n}(t) - \\ B_m\log_2 \left(
1 + \frac{G}{\sigma^2}\sum_{n \in S} \frac{p_{n}(t)}{(a_{nm}^2+\delta_{nm}^2 + q_{nm}(t)^2)^{\alpha}} \right) \leq 0.
\end{multline}
Since the non-negative weighted sum of functions preserves convexity properties, without loss of generality we take $S$ to be a singleton, and drop subscripts. We also drop time dependencies. The above function is the composition of two functions $\phi_1 \circ \phi_2(\cdot)$, respectively defined as 
\begin{align}
\phi_1(r,\phi_2(\cdot)) &\triangleq r - B \log_2(1+\phi_2(\cdot)), \\
\phi_2(p,q) &\triangleq \frac{G}{\sigma^2}\frac{p}{(a^2+\delta^2 + q^2)^\alpha}.
\end{align}
The function $(p,q) \mapsto \phi_2(p,q)$ is concave on the domain $\mathbb{R}_+ \times \mathbb{R}$. We show this by dropping constants and considering the simpler function $h(x,y)\triangleq xy^{-2\alpha}$ with Hessian
\begin{equation} 
	\nabla^2 h(x,y) = \begin{bmatrix}
	0 & \frac{-2\alpha}{y^{-2\alpha-1}} \\
   \frac{-2\alpha}{y^{-2\alpha-1}} & \frac{2\alpha(2\alpha-1)x}{y^{-2\alpha-2}}  \\
	\end{bmatrix},
\end{equation}
which is negative semi-definite, because it is symmetric with non-positive sub-determinants. Therefore, $\phi_2$ is jointly concave in both power and the difference in positions over the specified domain. $\phi_1$ is convex and non-increasing as a function of $\phi_2$. Since the composition of a convex, non-increasing function with a concave function is convex~\cite{boyd2004convex}, all data rate constraint functions are convex functions of~$(r,p,q)$. 
\end{proof}

The posynomial objective function is not convex over the whole of its domain and the logarithmic data rate term prevents the use of geometric programming (GP) methods.
\begin{lemma} \label{lemma:2}
The following problem
\begin{subequations} 
\begin{align}
& \min_{v_n, F_n} \int_0^T F_n(t) v_n(t) \mathrm{d}t  \\
 \text{s.t. } & \forall t \in \mathcal{T}  \notag \\
& F_n(t) - \Omega(v_n(t)) = m_n\dot{v}_n(t) \\
& F_{\text{min}} \leq f_m(t) \leq F_{\text{max}} \label{eq:vConst1} \\
& v_n(t) \geq 0 \label{eq:vConst2} \\
& v_n(0) = v_{n,\text{init}} \label{eq:vConst3}
\end{align}
\end{subequations}
of minimizing propulsion energy of a single node $U_n$, subject to initial and final conditions, admits an equivalent convex form for mappings $v_n(t) \mapsto \Omega(v_n(t))$ satisfying Assumption~\ref{ass:umption2} and force bounds $(F_\text{min},F_\text{max}) = (-\infty,\infty)$.
\end{lemma}

\begin{proof}
By noting that $F_n(t) = \Omega(v_n(t)) + m_n\dot{v}_n(t)$, we move the equality into the cost function, rewriting the problem as 
\begin{equation}
\min_{v_n} \phi(v_n)  \text{ s.t.\ (\ref{eq:vConst1})--
(\ref{eq:vConst3})}, 
\end{equation}
where
\begin{equation}
\phi(v_n) \triangleq \underbrace{\int_0^T v_n(t)\Omega(v_n(t)) \mathrm{d}t}_{\phi_1(v_n)} + \underbrace{\int_0^T v_n(t)\dot{v}_n(t)\mathrm{d}t}_{\phi_2(v_n)}.
\end{equation}

We now show that both $\phi_1(\cdot)$ and $\phi_2(\cdot)$ are convex. Starting with the latter, by performing a change of variable, the analytic cost is derived by first noting that $\phi_2(v_n)$ is the change in kinetic energy
\begin{equation}
\phi_2(v_n) = m_n\int_{v_n(0)}^{v_n(T)} v dv = \frac{m_n}{2}\left(v_n^2(T)-v_n^2(0)\right),
\end{equation}
which is a convex function of $v_n(T)$ subject to fixed initial conditions \eqref{eq:vConst2}; in fact, it is possible to drop the $v_n^2(0)$ term completely without affecting the argmin.
By Assumption~\ref{ass:umption2}, $v_n(t) \mapsto v_n(t)\Omega(v_n(t))$ is convex and continuous on the admissible domain of speeds. Since integrals preserve convexity, the total cost function $\phi(\cdot)$ is also convex. 

Removal of thrust $F$ as a decision variable results in the set
\begin{equation}
\mathcal{V}_F \triangleq \{v_n \mid F_\text{min} \leq \Omega(v_n(t)) + m_n \dot{v}_n(t) \leq F_\text{max} \}.
\end{equation}
Even if $\Omega(\cdot)$ is convex on the admissible range of speeds, the lower bound represents a concave constraint not admissible within a convex optimization framework.
Therefore, dropping  constraints on thrust results in a final convex formulation of
\begin{subequations} 
\begin{align}
& \min_{v_n} \int_0^T v_n(t)\Omega(v_n(t)) \mathrm{d}t +  \frac{m_n}{2}\left(v_n^2(T)-v_n^2(0)\right) \\
 &\text{s.t. }  \forall t \in \mathcal{T} \notag\\
&V_\text{min} \leq v_n \leq V_\text{max} \\ 
& v_n(0) = v_{n,\text{init}}.
\end{align}
\end{subequations}
\end{proof}
Addition of  bounds $v_n \in \mathcal{V}_F$ naturally results in a difference of convex (DC) problem \cite{yuille2003concave} that may be solved through exhaustive or heuristic procedures.

\begin{theorem}
In the absence of  constraints on thrust, the general problem (\ref{eq:GeneralProbForm}) admits an equivalent convex form.
\end{theorem}
\begin{proof}
Non-convexities in this formulation arise from the posynomial function of speed $v(t)$ and thrust~$F_m(t)$ in the cost function (\ref{eq:CostFunc}), the nonlinear force balance equality~\eqref{eq:AccelerationBound}, and the capacity region data rate constraints (\ref{eq:CapRegionConst}). The cost function is a superposition of the energies used by each node for propulsion and transmission. By noting that there is no coupling between nodes or between propulsion and transmission powers in this cost, the transformation used in Lemma~\ref{lemma:2} may be used to eliminate the nonlinear equality. We eliminate  $F_n(t)$ and $\dot{v}_n(t)$ and move the nonlinear equality into the objective function, simultaneously convexifying the objective to get
\begin{subequations}
\begin{align*}
   &  \Min_{p,r,s,q,v} \sum_{n=1}^N \left[ \int_0^T p_n(t) + v_n(t)\Omega(v_n(t))\mathrm{d}t + \frac{m_n}{2} v_n^2(T) \right] \\
   \text{s.t.} & \quad \forall n \in \mathcal{N},  m \in \{\mathcal{N},N+1\},  t \in \mathcal{T}, v \in \mathcal{V}^N,	\mathcal{S}\subseteq\mathcal{N} \notag \\
	& \text{(\ref{eq:CapRegionConst})--(\ref{eq:initVelocity}),\ \eqref{eq:defVel}},\   \tilde{Y}_{n,\text{min}} \leq \tilde{Y}_n(t) \leq \tilde{Y}_{n,\text{max}}
\end{align*}
\end{subequations}
where $\tilde{Y}_{n}(t) \triangleq \left(p_n(t),r_n(t),s_n(t),q_n(t),v_n(t)\right)$, and the bounds $\tilde{Y}_{n,\text{min}},$ and $ \tilde{Y}_{n,\text{max}}$ are similarly changed. It follows from Lemma~\ref{lemma:1} that all data rate constraints in (\ref{eq:CapRegionConst}) are also convex, therefore the whole problem is convex.
\end{proof}

\begin{table}[b!]
\begin{tabular}{ | c | c | c | c | c | c |}
 \hline
 $\sigma^2$ [W] & $B$ [Hz] & $M$ [GB]& $P_{\max}$ [W] & $\alpha$ & $T$ [min] \\
 \hline
 $10^{-10}$ & $10^5$ & $1$ & $100$ & 1.5 & 20 \\
 \hline
\end{tabular}
\caption{\label{tab:SimParameters} Dynamic model parameters that have been used across all simulation results.}
\end{table}

\section{Simulation Results} 	
The open source primal dual Interior Point solver \texttt{Ipopt v.3.12.4} has been used through the \texttt{MATLAB}
interface. 
 Table~\ref{tab:SimParameters} contains parameters common to the following experiments. Force constraints are relaxed in all experiments. 
From \cite{bekmezci2013flying}, the speed of a typical UAV is in the range $30$ to $460$ km/h. 
All nodes are initialized to their average speeds $v_{n,\text{init}}= (V_\text{max}+V_\text{min})/2$, We assume all nodes move in symmetric trajectories around the AP such that $Q_{n,\text{final}} = -Q_{n,\text{init}} = (T/2)v_{n,\text{init}}$.
 
\subsection{Single Node}
A single mobile node $U_1$ of mass~$3$\,kg traveling at fixed altitude $a=1000$\,m and lateral displacement $\delta=0$\,m, depicted in Figure~\ref{fig:TwoNodeDiagram}, is considered first. In this section, simulation results are presented for the problem of minimizing the total communication energy to offload all data to $U_0$. This is compared to a water-filling solution~\cite{wolf2011introduction} for minimizing the transmission energy. Subscripts denoting different nodes have been dropped in the remainder of this section. 
Specifically, we use $\Omega(\cdot)$ of the form
\begin{equation}
\Omega(x) \triangleq \left\{
\begin{array}{ll} 
\infty, \ &\forall x \in  (-\infty,0)\\
C_{D1}x^{2} + C_{D2}x^{-2}, \ &\forall x \in [0,\infty),
\end{array}
\right.
\end{equation}
where $C_{D1} = 9.26\times10^{-4}$ is the parasitic drag coefficient and $C_{D2}=2250$ is the lift induced drag coefficient \cite{zeng2016energy}.

Simulation results are shown in Figure~\ref{fig:SingleNodeCommsProblemSimulationResults} for a storage buffer initialized to $D=75$\,MB and speeds restricted in the range $[V_\text{min},V_\text{max}]=[30,100]$\,km/h. This results in a total energy expenditure of $309.50$\,kJ, where $105.05$\,kJ is due to transmission and $204.51$\,kJ is due to propulsion. Of this, only $48.01$\,kJ of \emph{extra} propulsion energy is used to vary speed on top of the base energy required to traverse the distance at a constant speed. 
Furthermore, the problem would have been infeasible if the node was restricted to a constant speed of $65$\,km/h. We note that, with the given parameterization, it is possible to transmit up to $78$\,MB of data in the defined time interval. 

In comparison, if the speed of $U_1$ is fixed, then the maximum transmittable data is approximately $56$\,MB, using~$120.00$\,kJ of transmission energy. 
Although considerably more energy is used, the optimal power policy for a fixed trajectory is characterized by a water-filling solution, an equivalent proof of which may be found in \cite{wolf2011introduction}. This problem results in a one dimensional search space, easily solved through such algorithms as binary search.

\begin{figure}[t!]
\centering
\vspace{0.5em}
\subfloat[Optimal transmission power and propulsion force used by $U_1$.]{
	\label{subfig:singelNodeSpeedVarying1}
	\includegraphics[width=\columnwidth]{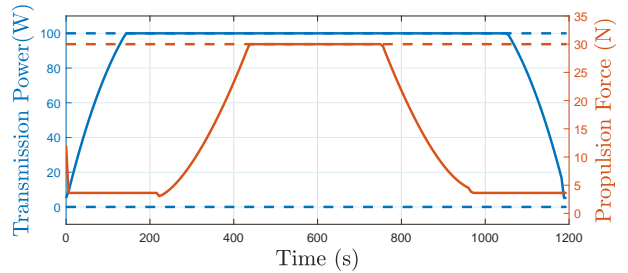}} 

\subfloat[Associated achieved data rate and velocity profile of $U_1$.]{
	\label{subfig:singelNodeSpeedVarying2}
	\includegraphics[width=\columnwidth]{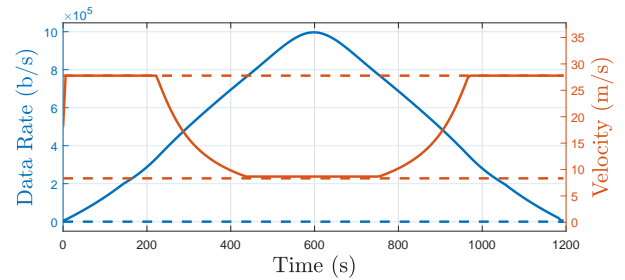}} 
   
\caption{Simulation results for the single-node problem, with trajectories shown as solid, and bounds shown as dashed lines.}
\label{fig:SingleNodeCommsProblemSimulationResults}
\end{figure}
\subsection{Multiple Nodes} \label{sec:MultNode}

We now investigate the transmission energy problem for two nodes, traveling in parallel trajectories at fixed speeds such that $V_\text{max}=V_\text{min}=65$\,km/h, as depicted by the green lines in Figure~\ref{fig:TwoNodeDiagram}. Relaying is not allowed, as may be the case if no bandwidth is allocated to $U_1$ and $U_2$ to receive each other's transmissions, equivalently turning them into pure source nodes. Simulation results are presented in  Figure~\ref{fig:TwoNodeSimulation}.

$U_1$ is closer to the AP at all times, and therefore is advantaged in that it experiences more favorable channel conditions. The disadvantaged node $U_2$ transmits for a longer duration due to the smaller relative change in its channel gain. The interior point algorithm converged after~42 iterations to a minimum energy of $52.707$kJ and $26.77$kJ for $U_1$ and $U_2$, respectively, for a starting data load of $D_1=D_2=25$\,MB. 

\begin{figure}[t!]
\centering
\vspace{0.5em}
\subfloat[Transmit powers of nodes $U_1$ and $U_2$.]{
	\label{subfig:error1}
	\includegraphics[width=\columnwidth]{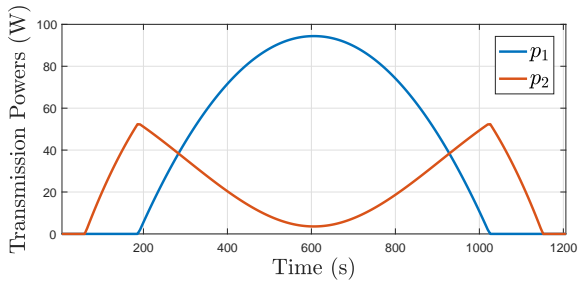}} 

\subfloat[Associated transmission rates achieved by nodes $U_1$ and $U_2$.]{
	\label{subfig:error2}
	\includegraphics[width=\columnwidth]{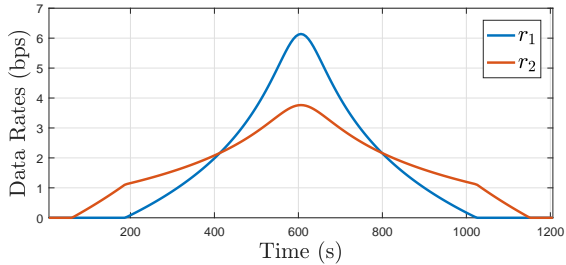} } 

\caption{Simulation results for the two-node transmission power problem.}
\label{fig:TwoNodeSimulation}
\end{figure}

It is notable that the advantaged node uses considerably more transmission energy than the disadvantaged node. Referring to \cite{cheng1993gaussian}, which derives two-user optimal power allocations that achieve arbitrary rate tuples on the boundary of $\mathcal{C}$ we explain this as follows. From Figure \ref{fig:CapacityRegion}, the optimal rate pairs for given transmit powers $p_1$ and $p_2$ lie on the segment $L_3$.  Equivalently, $ \exists \varrho \in [0,1]$ such that the rate pair for an arbitrary point $R^{(*)}=(r_1^{(*)},r_2^{(*)})$ on $L_3$ is given by the interpolation $R^{(*)} = \varrho \cdot R^{(1)} + (1-\varrho) \cdot R^{(2)}$.
We may interpret $\varrho$ as being the priority assigned to each transmitting node by the $U_0$ when SIC is being carried out. $\varrho=1$ means that data from $U_1$ is being decoded second, subject to a lower noise rate, while $\varrho=0$ means the opposite decoding order. We may think of the mapping $t \mapsto \varrho(t)$ as  a time-varying priority. 
However, by calculating $\varrho(t)$ from the optimum powers and rates seen in Figure~\ref{fig:TwoNodeSimulation}, we find that $\varrho(t)=0, \forall t \in \mathcal{T}$ such that $p_1(t)>0,p_2(t)>0$. In other words, the disadvantaged node is always given priority, which is why it uses less energy at the optimum, even though it always experiences a worse channel gain.

\section{Conclusions}
We have presented a general optimization framework for joint control of propulsion and transmission energy for single/multi-hop communication links in robotic networks.
The relaxation of transmission constraints to theoretic capacity bounds, with relatively mild assumptions on the mobility model, results in a nonlinear but convex OCP.
We showed that optimizing over a fixed path, as opposed to a fixed trajectory, increases the feasible starting data by at least~$30\%$ for just a single node.
For the fixed-trajectory two-node MAC simulation, the optimal solution  has been presented and analyzed. 
Immediate extensions of this work include higher fidelity models, and analysis of the relay network encompassed in problem (\ref{eq:GeneralProbForm}). Considering the overarching goal of real-time control, further developments will be closed-loop analysis of the control strategy, and consideration of the computational burden and energy expenditure  \cite{thammawichai2016optimizing, nazemi2016qoi} in the network.

\balance
\bibliography{CDC2017} 
\bibliographystyle{ieeetr}
\end{document}